\documentclass[11pt]{llncs}

\usepackage{amssymb}
\usepackage{amsmath}
\usepackage{amsfonts}
\usepackage[colorlinks=true, citecolor=blue]{hyperref}
\usepackage{color,graphicx,setspace}
\usepackage{geometry}
\usepackage{xcolor}

\usepackage{algorithm}
\usepackage[noend]{algpseudocode}

\def \eps   {{\varepsilon}}
\newtheorem{hypo}{Hypothesis}

\title{Scheduling Lower Bounds via\\ AND Subset Sum}
\author{Amir Abboud\inst{1}, Karl Bringmann\inst{2}\thanks{This work is part of the project TIPEA that has received funding from
the European Research Council (ERC) under the European Union’s Horizon 2020 research and innovation programme (grant agreement No. 850979).}, Danny Hermelin\inst{3}, Dvir Shabtay\inst{3}}

\institute{
IBM Almaden Research Center,\\ San Jose, California, USA\\
\email{amir.abboud@gmail.com}, 
\and
Saarland University and Max Planck Institute for
Informatics,\\ Saarland Informatics Campus, Saarbr\"ucken, Germany\\
\email{bringmann@cs.uni-saarland.de}, 
\and
Ben-Gurion University of the Negev,\\ Marcus Family Campus, Beer Sheva, Israel\\\email{hermelin@bgu.ac.il,dvirs@bgu.ac.il}
}

\date{September 2019}

\begin{document}

\sloppy
\maketitle

\begin{abstract}
Given $N$ instances $(X_1,t_1),\ldots,(X_N,t_N)$ of Subset Sum, the AND Subset Sum problem asks to determine whether all of these instances are yes-instances; that is, whether each set of integers $X_i$ has a subset that sums up to the target integer $t_i$. We prove that this problem cannot be solved in time $\widetilde{O}((N \cdot t_{max})^{1-\eps})$, for $t_{max}=\max_i t_i$ and any $\eps > 0$, assuming the $\forall \exists$ Strong Exponential Time Hypothesis ($\forall \exists$-SETH). We then use this result to exclude $\widetilde{O}(n+P_{max} \cdot n^{1-\eps})$-time algorithms for several scheduling problems on $n$ jobs with maximum processing time $P_{max}$, assuming $\forall \exists$-SETH. These include classical problems such as $1||\sum w_jU_j$, the problem of minimizing the total weight of tardy jobs on a single machine, and $P_2||\sum U_j$, the problem of minimizing the number of tardy jobs on two identical parallel machines.  
\end{abstract}



\section{Introduction}

The Subset Sum problem is one of the most fundamental problems in computer science and mathematics: Given $n$ integers $X=\{x_1,\ldots, x_n\} \subset \mathbb{N}$, and a target value $t \in \mathbb{N}$, determine whether there is a subset of $X$ that sums\footnote{Note that we can ignore any numbers $x_i > t$, so we will assume throughout the paper that $\max(X) \le t$.} to $t$. 
This problem appeared in Karp's initial
list of 21 NP-complete problems~\cite{Karp72}, and entire books have been devoted to it and to its closely related variants~\cite{KPP04book,MT90book}. Most relevant to this paper is the particular role Subset Sum plays in showing hardness for various problems on integers, essentially being the most basic  such problem where hardness arises exclusively from the additive nature of the problem. In particular, in areas such as operations research, Subset Sum plays a similar role to that of 3-SAT, serving as the core problem used in the vast majority of reductions (see \emph{e.g.}~\cite{Brucker2006,ChengSN16,DuL90,Karp72,KovalyovPesch10,Pinedo2008}). 
Many important problems can be shown to be generalizations of Subset Sum (by easy reductions) including scheduling problems, Knapsack, and Bicriteria Shortest Path. The broad goal of this paper is to understand the fine-grained complexity of such important problems, and more specifically whether the complexity of such generalizations is the same as that of Subset Sum or higher.

While Subset Sum (and its generalizations) is NP-hard, it is well-known that it can be solved in pseudo-polynomial time $O(t \cdot n)$ with the classical dynamic programming algorithm of Bellman \cite{bellman1957dynamic}.
Much more recently, this upper bound was improved to $\widetilde{O}(t+n)$~\cite{Bring17,jin2018simple,KoiliarisX17}; this is a significant improvement in the dense regime of the problem, \emph{e.g.} if $t=O(n^2)$ the new algorithms achieve quadratic as opposed to cubic time.
Most recently, in the dense regime the fine-grained complexity of Subset Sum was essentially resolved under the Strong Exponential Time Hypothesis (SETH) by the authors of this paper~\cite{AbboudBHS17} (the same lower bound was previously known under the incomparable Set Cover Conjecture~\cite{Cygan+16}).
SETH \cite{IP2001,IPZ2001} postulates that there is no $O(2^{(1-\eps)n})$-time algorithm for deciding the satisfiability of a $k$-CNF formula, for some $\eps > 0$ independent of~$k$. 
\begin{theorem}[Hardness of Subset Sum~\cite{AbboudBHS17}]
\label{thm:SubsetSum}%
Assuming SETH, there is no $\eps > 0$ and $\delta < 1$  such that Subset Sum on $n$ numbers and target $t$ can be solved in time $O(t^{1-\eps} \cdot 2^{\delta n})$. 
\end{theorem}

The lower bound given by Theorem~\ref{thm:SubsetSum} translates directly to several generalizations of Subset Sum, but
does this yield tight lower bounds for the generalizations? Or can we prove higher lower bound for them? To answer this kind of question, the OR Subset Sum problem was introduced in~\cite{AbboudBHS17}: Given $N$ instances $(X_1,t_1),\ldots, (X_N,t_N)$ of Subset Sum, determine whether at least one of these instances is a yes-instance; that is, whether there exists an $i \in \{1,\ldots,N\}$ such that $X_i$ contains a subset that sums up to~$t_i$. While it seems natural to assume that no algorithm can solve this problem faster than solving each of the $N$ Subset Sum instances independently, it is not clear how to prove this. In fact, an $O(N^{1/10} \cdot \max_i t_i)$ time algorithm for this problem does not directly break the lower bound for Subset Sum. Nevertheless, one can still show a tight lower bound by taking a somewhat indirect route: SAT does have a reduction to its OR variant, and then Theorem~\ref{thm:SubsetSum} allows us to reduce OR SAT to OR Subset Sum. 
\begin{theorem}[Hardness of OR Subset Sum~\cite{AbboudBHS17}]
\label{thm:ORSubsetSum}%
Assuming SETH, there are no $\eps,\delta > 0$ such that there is an $O(N^{1+\delta-\eps})$ time algorithm for the following problem: Given $N$ Subset Sum instances, each with $O_{\delta,\eps}(\lg N)$ integers and target $O(N^\delta)$, determine whether one of these instances is a yes-instances. 
\end{theorem}

Thus, while Subset Sum admits\footnote{The term $\widetilde{O}()$ is used here and throughout the paper to suppress logarithmic factors.} $\widetilde{O}(n+t)$-time algorithms~\cite{Bring17,jin2018simple,KoiliarisX17}, SETH rules time $\widetilde{O}(N+t)$ for OR Subset Sum. For example, when $N=O(n)$ and $t=O(n^2)$, Subset Sum can be solved in time $O(n^2)$, but OR Subset Sum has a cubic lower bound according to the above theorem. This distinction was used in~\cite{AbboudBHS17} to show a higher lower bound for a generalization of Subset Sum that is a particularly  prominent problem in the operations research community, the Bicriteria Shortest Path problem~\cite{GarroppoGT10,YounisF03}: 
Given a graph~$G$ with edge lengths and edge costs, two vertices $s$ and $t$, and a budget $B$, determine whether there is an $s,t$-path of total length at most $B$ and total cost at most $B$. While Theorem~\ref{thm:SubsetSum} immediately rules out time $B^{1-\eps} \cdot 2^{o(n)}$, it leaves open the possibility of an $\widetilde{O}(B+n)$ algorithm (as is known to exist for Subset Sum). As it turns out, Bicriteria Shortest Path can not only encode a \emph{single} Subset Sum instance, but even \emph{several} instances, and thus Theorem~\ref{thm:ORSubsetSum} yields an $\Omega(n+Bn^{1-\eps})$ lower bound under SETH.  

\subsection{An Analogue of Theorem~\ref{thm:ORSubsetSum} for AND Subset Sum}

While the OR variant in Theorem~\ref{thm:ORSubsetSum} is perfectly suited for showing lower bounds for Bicriteria Shortest Path and other problems of a similar type, there are others, such as the scheduling problems discussed below, whose type can only capture an AND variant: Given $N$ instances of Subset Sum, determine whether \emph{all} are yes-instances. It is natural to wonder whether there is a fine-grained reduction from SAT to AND Subset Sum (either directly or indirectly, by first reducing to AND SAT). Intuitively, the issue is that SAT, Subset Sum, and their OR variants have an $\exists$ quantifier type, while AND SAT and AND Subset Sum have a $\forall \exists$ quantifier type. Reducing one type to another seems very challenging, but fortunately, a morally similar challenge had been encountered before in fine-grained complexity and resolved to some extent as follows.

First, we can observe that the reduction we are looking for is impossible under the Nondeterministic Strong Exponential Time Hypothesis (NSETH)~\cite{CarmosinoGIMPS16} which states that no non-deterministic $O(2^{(1-\eps)n})$-time algorithm can decide whether a given $k$-CNF is unsatisfiable, 
for an $\eps>0$ independent of $k$. This hypothesis was introduced to show non-reducibility results. Intuitively, NSETH says that even though SAT is easy for nondeterministic algorithms its complement is not. Therefore, if for a certain problem both it and its complement are easy for nondeterministic algorithms then a reduction from SAT is impossible. Note that AND SAT, AND Subset Sum, and their complements admit efficient nondeterministic algorithms: to prove that the AND is ``yes'' we can guess a solution in each instance, and (for the complement) to prove that the AND is ``no'' we can guess the index of the instances that is ``no''. (Notice that the latter is not possible for the OR variants.)

There are already conjectures in fine-grained complexity that can capture problems with a $\forall \exists$ type.
In the ``$n^2$ regime'', where SAT is faithfully represented by the Orthogonal Vectors (OV) problem\footnote{Given two sets of $n$ binary vectors of dimension $O(\log{n})$, decide whether there is a vector in the first set and a vector of the second set that are orthogonal. SETH implies that this problem cannot be solved in time $O(n^{2-\eps})$ \cite{williams2005new}, and essentially all SETH-based $n^2$ lower bounds go through this problem.} which has an $\exists$ type, Abboud, Vassilevska Williams and Wang \cite{abboud2016approximation} introduced a hardness hypothesis about the Hitting Set (HS) problem\footnote{Given two sets of $n$ binary vectors of dimension $O(\log{n})$, decide whether for all vectors in the first set there is an orthogonal vector in the second set. The Hitting Set Hypothesis states that this problem cannot be solved in time $O(n^{2-\eps})$ for any $\eps > 0$.} which is the natural $\forall \exists$ type variant of OV.
This hypothesis was used to derive lower bounds that cannot (under NSETH) be based on OV or SETH, e.g. for graph median and radius \cite{abboud2016approximation,ancona2018algorithms,demaine2017fine} and for Earth Mover Distance \cite{rohatgi2019conditional}, and was also studied in the context of model checking problems~\cite{gao2018completeness}.
Going back to the ``$2^n$ regime'', the analogous hypothesis, which implies the HS hypothesis, is the following.

\begin{hypo}[$\forall \exists$-SETH]
There is no $0 < \alpha  < 1$ and $\eps > 0$ such that for all $k \geq 3$ we can decide in time $O(2^{(1 - \eps)n})$, given a $k$-CNF formula $\phi$ on $n$ variables $x_1,\ldots,x_n$, whether for all assignments to $x_1,\ldots,x_{\lceil \alpha\cdot n \rceil}$ there exists an assignment to the rest of the variables that satisfies $\phi$, that is, whether: 
\[ \forall x_1,\ldots,x_{\lceil \alpha\cdot n \rceil} \exists x_{\lceil \alpha\cdot n \rceil+1},\ldots, x_n: \phi(x_1,\ldots,x_n) = \text{true}. \]
\end{hypo}

Note that this hypothesis may also be thought of as the $\Pi_2$-SETH, where $\Pi_2$ is the second level of the polynomial hierarchy, and one can also think of higher levels of the polynomial hierarchy. Indeed, Bringmann and Chaudhury \cite{BringmannC19} recently proposed such a version, called Quantified-SETH, in which we can have any constant number $q\geq 1$ of alternating quantifier blocks, with a constant fraction of the variables in each block\footnote{However, we remark that for the purposes of their paper as well as ours $\forall\exists$-SETH is sufficient; Quantified-SETH is merely mentioned for inspiration. They were motivated by understanding the complexity of the polyline simplification problem from geometry (which turns out to have a $\forall \forall \exists$ type).}.
Non-trivial algorithms for Quantified-SAT exist~\cite{santhanam2014beating}, but none of them can refute even the stronger of these hypotheses. 

It is important to note that while $\forall \exists$-SAT is a strictly harder problem than SAT (as adding more quantifiers can only make the problem harder), in the restricted setting of $\forall\exists$-SETH, where there is a constant fraction of the variables in each quantifier block, the situation is the opposite! 
A faster algorithm for SAT does imply a faster algorithm for $\forall \exists$-SAT: exhaustively search over all assignments to the universally quantified $\alpha n$ variables and for each assignment solve SAT on $(1-\alpha)n$ variables. A reduction in the other direction is impossible under NSETH\footnote{This is analogous to the ``$n^2$ regime'' where HS implies OV but not the other way, assuming NSETH.}.
 Therefore, $\forall \exists$-SETH is a stronger assumption than SETH, which explains why it is helpful for proving more lower bounds, yet it seems equally plausible (to us).
In particular, it gives us a tight lower bound for AND Subset Sum which we will use to show higher lower bounds for scheduling problems.
\begin{theorem}[Hardness of AND Subset Sum]
\label{thm:ANDSubsetSum}%
Assuming $\forall \exists$-SETH, there are no $\eps,\delta > 0$ such that the following problem can be solved in time $O(N^{1+\delta-\eps})$: Given $N$ Subset Sum instances, each with $O(N^\eps)$ integers and target $O(N^\delta)$, determine whether all of these instances are yes-instances. 
%
\end{theorem}

Note that in comparison with the OR Subset Sum case (Theorem~\ref{thm:ORSubsetSum}), the size of our instances is polynomial $O(N^\eps)$ instead of logarithmic $O_{\delta,\eps}(\log N)$. We leave it as an open problem whether this is inherent or Theorem~\ref{thm:ANDSubsetSum} can be improved.

It follows from Theorem~\ref{thm:ANDSubsetSum} that AND Subset Sum on $N$ instances, each on at most $s$ integers and with target at most $t$, cannot be solved in time $\widetilde{O}(Ns + t (Ns)^{1-\eps})$. We show that the same holds for the Partition problem, which is the special case of Subset Sum where the target is half of the total input sum. This is the starting point for our reductions in the next section.

\begin{corollary} \label{cor:thecor}
Assuming $\forall \exists$-SETH, there is no $\eps > 0$ such that the following problem can be solved in time $\widetilde{O}(Ns + t (Ns)^{1-\eps})$: Given $N$ Partition instances, each with at most $s$ integers and target at most $t$, determine whether all of these instances are yes-instances. 
\end{corollary}

\subsection{Scheduling lower bounds}

To exemplify the power of Theorem~\ref{thm:ANDSubsetSum}, we use it to show strong lower bounds for several non-preemptive scheduling problems that generalize Subset Sum. These problems include some of the most basic ones such as minimizing the total weight of tardy jobs on a single machine, or minimizing the number of tardy jobs on two parallel machines. Theorem~\ref{thm:lowerbounds} below lists all of these problems; they are formally defined in Section~\ref{sec:SLB} and each requires a different reduction. To describe the significance of our new lower bounds more clearly, let us focus on only one of these problems, $P_2||\sum U_j$, for the rest of this section.
The input to this problem is a set of $n$ jobs, where each job $J_j$ has a processing time $p_j$ and a due date $d_j$, and the goal is to schedule all jobs on two parallel machines so that the number of jobs exceeding their due dates is minimal. Let $P=\sum_j p_j$ and $P_{max} = \max_j p_j$ denote the sum of processing times and maximum processing time of the input jobs. Observe that $P \leq P_{max} \cdot n$.

The standard dynamic programming algorithm for this problem runs in $O(P \cdot n)=O(P_{max} \cdot n^2)$ time~\cite{LawlerMoore}, and it is not known whether this running time is the best possible. Nevertheless, there is a well-known easy reduction from Subset Sum on numbers $x_1,\ldots,x_n$ to $P_2||\sum U_j$ that generates an instance with total processing time $P=\sum {x_i} = O(n \cdot t)$ and $P_{max} = \max{x_i} = O(t)$. Thus, using Theorem~\ref{thm:SubsetSum}, we can rule out $P^{1-\eps}\cdot 2^{o(n)}$-time and $P_{max}^{1-\eps}\cdot 2^{o(n)}$-time algorithms for $P_2||\sum U_j$. However, this leaves open the possibility of $\widetilde{O}(P_{max} + n)$-time algorithms, which would be near-linear as opposed to the currently known cubic algorithm in a setting where $P_{max}=\Theta(n)$ and $P=\Theta(n^2)$.
One approach for excluding such an upper bound is to first prove the impossibility of an algorithm for Subset Sum with running time $\widetilde{O}(\max_{x \in X}x +n)$. 
However, such a result has been elusive and is perhaps the most interesting open question in this context~\cite{AxiotisBJTW19,EisenbrandW18,GalilMargalit91,KoiliarisX17,Pisinger99}.
Instead, taking an indirect route, we are able to exclude such algorithms with an $\Omega(n+P_{max}n^{1-\eps})$ lower bound under $\forall\exists$-SETH by showing that $P_2||\sum U_j$ can actually encode the AND of several Subset Sum instances. 
In particular, in the above regime we improve the lower bound from linear to quadratic.


\begin{theorem}
\label{thm:lowerbounds}
Assuming $\forall\exists$-SETH, for all $\eps>0$, none of the following problems have $\widetilde{O}(n +P_{max} \cdot n^{1-\eps})$-time algorithms: 
\begin{itemize}
\item $1||\sum w_jU_j$,\; $1|Rej\leq R|\sum U_j$,\; $1|Rej \leq R|T_{max}$,\; and \; $1|r_j \geq 0, Rej \leq R|C_{max}$.
\item $P_2||T_{max}$,\; $P_2||\sum U_j$,\; $P_2|r_j|C_{max}$,\;
and \; $P_2|\text{level-order}|C_{max}$. 
\end{itemize}
\end{theorem}

All problems listed in this theorem are direct generalizations of Subset Sum, and each one admits a $O(P \cdot n) = O(P_{max} \cdot n^2)$-time algorithm via dynamic programming~\cite{LawlerMoore,Rothkopf1966,ShabtayGK13}. 

We note that the distinction between running times depending on $P$ versus $P_{max}$ and~$n$ relates to instances with low or high variance in their job processing times. In several experimental studies, it has been reported by researchers that the ability of scheduling algorithms to solve NP-hard problems deteriorates when the variance in job processing time increases (see \emph{e.g.}~\cite{Khowala2008,PanShi07,Ragatz93}). Our results  provide theoretical evidence for this claim by showing tighter lower bounds on the time complexity of several scheduling problems based on the maximum processing time~$P_{max}$.


\section{Quantified SETH Hardness of AND Subset Sum}

In the following we provide a proof for Theorem~\ref{thm:ANDSubsetSum}, the main technical result of the paper. For this, we present a reduction from Quantified $k$-SAT to AND Subset Sum which consists of two main steps. The first step uses a tool presented in~\cite{AbboudBHS17} which takes a (non-quantified) $k$-SAT instance and reduces it to subexponentially many Subset Sum instances that have relatively small targets. The second step is a new tool, which we develop in Section~\ref{subsec:OR2One}, that takes many Subset Sum instances and reduces them to a single instance with only a relatively small increase of the output target.   

\subsection{Main construction}

The following two theorems formally state the two main tools that are used in our construction. Note that for our purpose, the important property here is the manageable increase of the output target in both theorems. The proof of Theorem~\ref{thm:redsatsubsetsum} can be found in~\cite{AbboudBHS17}, while the proof of Theorem~\ref{thm:MainReduction} is given in Section~\ref{subsec:OR2One}.

\begin{theorem}[\cite{AbboudBHS17}]
\label{thm:redsatsubsetsum}%
For any $\eps > 0$ and $k \geq 3$, 
given a $k$-SAT formula $\phi$ on $n$ variables, we can in time $2^{\eps n}\cdot n^{O(1)}$ construct $2^{\eps n}$ Subset Sum instances, each with  $O(n)$ integers and target at most $2^{(1+\eps)n}$, such that $\phi$ is satisfiable if and only if at least one of the Subset Sum instances is a yes-instance. 
\end{theorem}

\begin{theorem}
\label{thm:MainReduction}%
Given Subset Sum instances $(X_1,t_1), \ldots, (X_N,t_N)$, denoting $n = \max_i |X_i|$ and $t = \max_i t_i$, we can construct in time $(nN \log t)^{O(1)}$ a single Subset Sum instance $(X_0,t_0)$, with $|X_0| = O(nN)$ and $t_0 = t \cdot (nN)^{O(1)}$, such that $(X_0,t_0)$ is a yes-instance if and only if $(X_i,t_i)$ is a yes-instance for some $i \in \{1,\ldots,N\}$.
%
\end{theorem}

Using the two results above, the proof of Theorem~\ref{thm:ANDSubsetSum} follows by combining both constructions given by the theorems:

\begin{proof}[of Theorem~\ref{thm:ANDSubsetSum}]
Let $\phi$ be a $k$-SAT formula on $n$ variables and let $0 < \alpha < 1$.
We write $n_1 = \lfloor \alpha \cdot n \rfloor$ and $n_2 = n-n_1$, so that $n_1 \le \alpha n$ and $n_2 \le (1-\alpha)n+1$.
Our goal is to determine whether $\forall x_1,\ldots,x_{n_1} \exists x_{n_1+1},\ldots,x_n\colon \phi(x_1,\ldots,x_n)$ is true. 

We enumerate all assignments $\partial$ of the variables $x_1,\ldots,x_{n_1}$, and let $\phi_\partial$ be the resulting $k$-SAT formula on $n_2$ variables after applying $\partial$. Note that there are $2^{n_1}$ formulas $\phi_\partial$.

For each formula $\phi_\partial$, we run the reduction from Theorem~\ref{thm:redsatsubsetsum} with parameter~$\eps_0$, resulting in a set $\mathcal{I}_\partial$ of at most $2^{\eps_0 n_2}$ Subset Sum instances such that $\phi_\partial$ is satisfiable if and only if at least one of the instances in $\mathcal{I}_\partial$ is a yes-instance.
Note that each Subset Sum instance in $\mathcal{I}_\partial$ consists of $O(n_2)=O(n)$ integers and has target at most $t=2^{(1+\eps_0) n_2}$. Moreover, running this reduction for all formulas $\phi_\partial$ takes time $2^{n_1 + \eps_0 n_2} n^{O(1)}$.

Next, using Theorem~\ref{thm:MainReduction}, we reduce $\mathcal{I}_\partial$ to a single Subset Sum instance $(X_\partial,t_\partial)$ such that $(X_\partial,t_\partial)$ is yes-instance if and only if $\phi_\partial$ is a yes-instance, and so $\phi$ is a yes-instance if and only if \emph{all} $(X_\partial,t_\partial)$ are yes-instances. 
Note that we have $|X_\partial|=O(n \cdot 2^{\eps_0 n_2})$ and $t_\partial =  O(2^{(1+\eps_0) n_2} \cdot (n \cdot 2^{\eps_0 n_2})^\gamma)$ for some constant $\gamma > 0$ that replaces a hidden constant in Theorem~\ref{thm:MainReduction}. 
Moreover, running this step for all formulas $\phi_\partial$ takes time $O( 2^{n_1} \cdot (n 2^{\eps_0 n_2})^\gamma )$, where again $\gamma > 0$ replaces a hidden constant in Theorem~\ref{thm:MainReduction}.

Finally, we assume that for some $\eps',\delta > 0$ we can solve AND Subset Sum on $N$ instances, each with $O(N^{\eps'})$ integers and target $O(N^\delta)$, in time $O(N^{1+\delta-\eps'})$. 
Set $\eps := \eps' / (1+\delta) $ and note that then we can in particular solve AND Subset Sum on $N$ instances, each with $O(N^\eps)$ integers and target $O(N^\delta)$, in time $O(N^{(1+\delta)(1-\eps)})$; for convenience we use this formulation in the following. 
Clearly, we can assume $\eps < 1$.
Set 
\[ N := 2^{(1+\eps)n_1} = \Theta(2^{(1+\eps)\alpha n}),
\]
and note that the number of instances $(X_\partial,t_\partial)$ is $2^{n_1}\leq N$.
In order to apply the assumed algorithm to the instances $(X_\partial,t_\partial)$, we need to verify that $|X_\partial| = O(N^\eps)$ and $t_\partial = O(N^\delta)$. 
To this end, we set $\alpha := 1 / (1+\delta)$ and $\eps_0 := \min\{ \eps \alpha, \eps/(1+2\gamma) \}$, and check that
\begin{align*}
  |X_\partial| &= O(n \cdot 2^{\eps_0 n_2}) = O(2^{\eps_0 n}) = O(N^\eps), \\
  t_\partial &= O(2^{(1+\eps_0) n_2} \cdot (n \cdot 2^{\eps_0 n_2})^\gamma) = O(2^{(1+\eps_0(1+2\gamma))n_2}). 
\end{align*}
Using $\eps_0 \le \eps/(1+2\gamma)$ and $n_2 \le (1-\alpha)n+1$, we further simplify this to $t = O(2^{(1+\eps)(1-\alpha)n})$. From our setting of $\alpha = 1 / (1+\delta)$ it now follows that $t = O(2^{(1+\eps) \delta \alpha n}) = O(N^\delta)$.
Hence, we showed that the assumed algorithm for AND Subset Sum is applicable to the at most $N$ instances $(X_\partial,t_\partial)$. This algorithm runs in time
\[ O(N^{(1+\delta)(1-\eps)}) = O(2^{(1+\delta)(1-\eps)(1+\eps)\alpha n}) = O(2^{(1-\eps^2)n}). \]
Additionally, as analyzed above, the running time incurred by the reduction is bounded by
\begin{align*}
O\big(2^{n_1 + \eps_0 n_2} n^{O(1)} + 2^{n_1} \cdot (n 2^{\eps_0 n_2})^\gamma \big) = O(2^{n_1 + \eps_0(1+2\gamma)n_2}) &= O(2^{\alpha n + \eps (1-\alpha)n})  \\
&= O(2^{(1 - (1-\eps)(1-\alpha))n}).
\end{align*}
Hence, we can solve the quantified $k$-CNF formula $\phi$ in time $O(2^{(1 - \min\{\eps^2, (1-\eps)(1-\alpha)\})n})$, which for sufficiently large $k$ violates $\forall \exists$-SETH. \qed
\end{proof}


Next, we infer Corollary~\ref{cor:thecor} from Theorem~\ref{thm:ANDSubsetSum}. 

\begin{proof}[of Corollary~\ref{cor:thecor}]
Fix any $\delta > 0$, and let $\eps > 0$ to be chosen later.
For given Subset Sum instances $(X_1,t_1),\ldots,(X_N,t_N)$, each with $O(N^\eps)$ integers and target $O(N^\delta)$, our goal is to determine whether all of these instances are yes-instances.

For each $i$, we construct a Partition instance $(X^*_i,t^*_i)$ by setting
\[ X^*_i := X_i \cup \Big\{ \Big( \sum_{x \in X_i} x \Big) + t_i, 2 \cdot \Big(\sum_{x \in X_i} x \Big) - t_i \Big\} \qquad\text{and} \qquad t^*_i := \frac 12 \sum_{x \in X^*_i} x. \]
It is easy to see that the Partition instance $(X^*_i,t^*_i)$ is equivalent to the Subset Sum instance $(X_i,t_i)$. Indeed, the two additional items cannot be put on the same side of the partition, as their sum is too large. Putting them on different sides of the partition, it remains to split $X_i$ into a subset $Y_i \subseteq X_i$ summing to $t_i$ and the remainder $X_i \setminus Y_i$ summing to $(\sum_{x \in X_i} x) - t_i$, to obtain a balanced partition. 

Observe that $|X^*_i| = O(|X_i|) = O(N^\eps)$ and $t^*_i = O(|X_i| \cdot t_i) = O(N^{\delta+\eps})$. 

Now assume that we can solve AND Partition on $N$ instances, each with at most $s$ integers and target at most $t$, in time $O(Ns + t(Ns)^{1-\eps_0})$ for some $\eps_0 > 0$.
On the instances $(X^*_1,t^*_1),\ldots,(X^*_n,t^*_N)$, this algorithm would run in time 
\[ \widetilde{O}(Ns + t(Ns)^{1-\eps_0}) = \widetilde{O}(N^{1+\eps} + N^{\delta+\eps+(1+\eps)(1-\eps_0)}) = \widetilde{O}(N^{1+\eps} + N^{1 + \delta + 2\eps - \eps_0}). \]
Finally, we pick $\eps := \min\{\delta/2, \eps_0/3\}$ to bound this running time by $O(N^{1+\delta-\eps})$. This violates Theorem~\ref{thm:ANDSubsetSum}. \qed
\end{proof}

\subsection{From OR Subset Sum to Subset Sum}
\label{subsec:OR2One}%

We next provide a proof of Theorem~\ref{thm:MainReduction}, the second tool used in our reduction from Quantified $k$-SAT to Subset Sum. We will use the notion of average-free sets.

\begin{definition}[$m$-average-free set]
A set of integers $S$ is called $m$-average-free if for all (not necessarily distinct) integers $s_1,\ldots,s_{m+1} \in S$ we have:
\[
s_1 + \cdots +s_m = m \cdot s_{m+1} 
\,\,\text{ implies that }\,\,
s_1 = \cdots = s_{m+1}.
\]
\end{definition}

\begin{lemma}[\cite{Behrend47}]
\label{lem:average-free}%
Given $m \geq 2$, $M \geq 1$, and  $0 < \eps < 1$, an $m$-average-free set $S$ of size~$M$ with $S \subseteq [0, m^{O(1/\eps)} M^{1+\eps}]$ can be constructed in $M^{O(1)}$ time.
\end{lemma}

\begin{proof}[of Theorem~\ref{thm:MainReduction}]
Let $(X_1,t_1),\ldots,(X_N,t_N)$ be $N$ given Subset Sum instances, and write $t = \max_i t_i$ and $n = \max_i |X_i|$. We begin by slightly modifying these instances. 
First, let $t^* = (n+1) t$, and add to each $X_i$ the integer $t^*-t_i$. Clearly, there is a subset of $X_i$ which sums up to $t_i$ if and only if there is a subset of $X_i \cup \{t^*-t_i\}$ that sums up to $t^*$. 
Next, we add at most $2(n+1)$ copies of 0 to each instance, ensuring that all instances have the same number of integers $2(n+1)$, and that any instance which has a solution also has one which includes exactly $n+1$ integers. 
Note that these modifications only change $n$ by a constant factor, and $t$ by a factor $O(n)$, which are negligible for the theorem statement.

Therefore, with slight abuse of notation, henceforth we assume that we are given $N$ Subset Sum instances $(X_1,t_1),\ldots,(X_N,t_N)$ with $t_1=\ldots=t_N=t$ and $|X_1|=\ldots=|X_N|=2n$. Moreover, for any $i$ if there exists a subset $Y_i \subseteq X_i$ that sums up to $t$ then we can assume without loss of generality that $|Y_i| = n$.

We construct an $n$-average-free set $S=\{s_1,\ldots,s_N\}$, with $S \subseteq [0,N^2 \cdot n^{O(1)}]$, using Lemma~\ref{lem:average-free}. Let $S_{\textup{max}}=\max_{s \in S} s$. 

We are now ready to describe our construction of $(X_0,t_0)$. It will be convenient to view the integers in $(X_0,t_0)$ as binary encoded numbers, or binary strings, and to describe how they are constructed in terms of \emph{blocks} of consecutive bits. Each integer will consist of seven blocks of fixed sizes. Starting with the least significant bit, the first block has $\lceil \lg t \rceil$ bits and is referred to as the \emph{encoding block}, the third block has $\lceil \lg n \rceil$ bits and is referred to as the \emph{counting block}, the fifth block has $\lceil \log(n \cdot S_{\textup{max}}) \rceil = O(\log (nN))$ bits and is referred to as the \emph{verification block}, and the last block consists of a single bit. In between these blocks are blocks containing $\lceil \log(2nN) \rceil$ bits of value 0, whose sole purpose is to avoid overflows. 

For each integer $x_{i,j} \in X_i$, we construct a corresponding integer $x^0_{i,j} \in X_0$ as follows (here the `$|$'-characters are used only to differentiate between blocks, and have no other meaning):
\[
x^0_{i,j} \,\,\, = \,\,\,  0 \,\,\,|\,\,\, 
0 \cdots 0 \,\,\,|\,\,\, s_i  \,\,\,|\,\,\, 
0 \cdots 0 \,\,\,|\,\,\, 0 \cdots 01 \,\,\,|\,\,\, 
0 \cdots 0 \,\,\,|\,\,\, x_{i,j},
\]
Additionally, for each $i \in \{1,\ldots,N\}$ we construct an integer $x^0_i \in X_0$ associated with the instance $(X_i,t_i) $ as
\[
x^0_{i} \,\,\, = \,\,\,  1  \,\,\,|\,\,\, 
0 \cdots 0 \,\,\,|\,\,\, n \cdot (S_{\textup{max}}- s_i)  \,\,\,|\,\,\, 
0 \cdots 0 \,\,\,|\,\,\, 0 \cdots 0 \,\,\,|\,\,\,
0 \cdots 0 \,\,\,|\,\,\,  0.
\]
The two sets of integers described above constitute $X_0$. To complete the construction of the output instance, we construct the target integer $t_0$ as
\[
t_0 \,\,\, = \,\,\,  1 \,\,\,|\,\,\, 
0 \cdots 0 \,\,\,|\,\,\, n \cdot S_{\textup{max}}  \,\,\,|\,\,\, 
0 \cdots 0 \,\,\,|\,\,\, n \,\,\,|\,\,\, 
0 \cdots 0 \,\,\,|\,\,\,  t^*.
\]

Note that $|X_0|=O(\sum_i |X_i|) = O(nN)$ and $t_0 = t \cdot (nN)^{O(1)}$, as required by the theorem statement. Furthermore, the time required to construct $(X_0,t_0)$ is $(nN \log t)^{O(1)}$. 

We next argue that $(X_0,t_0)$ is a yes-instance if and only if $(X_i,t_i)$ is a yes-instance for some $i \in \{1,\ldots,N\}$. Suppose that there exists some $i \in \{1,\ldots,N\}$ and some $Y_i \subseteq X_i$ for which $\sum_{x_{i,j} \in Y_i} x_{i,j} = t$. By the discussion at the beginning of this proof, we can assume that $|Y_i| = n$. It is not difficult to verify that all integers in $Y_0 := \{x^0_{i,j} : x_{i,j} \in Y_i \} \cup \{x^0_i\}$ sum up to $t_0$. Indeed, by construction, the bits in the encoding block of these integers sum up to $\sum_{x_{i,j} \in Y_i} x_{i,j} = t$, the bits in the counting block sum up to $n$, the bits in the verification block sum up to $n \cdot S_{\textup{max}}$, and the last bit sums up to $1$. 

Conversely, assume that there is some subset $Y_0 \subseteq X_0$ with $\Sigma(Y_0)=\sum_{x \in Y_0} x = t_0$. Let $y_1,\ldots,y_{m} \in Y_0$ denote all integers of the form $x^0_{i,j}$ in $Y_0$, and let $x_{i_1,j_1},\ldots,x_{i_{m},j_{m}}\in X_1\cup \cdots \cup X_M$ denote the integers that appear in the encoding blocks of $y_1,\ldots,y_{m}$. Observe that as $m \leq 2nM$, by our construction the highest bit in each overflow block of $\Sigma(Y_0)$ must be 0. It follows that we can argue in each of the encoding block, counting block, verification block, and last block separately. This yields:
\begin{itemize}
\item $\sum_\ell x_{i_\ell,j_\ell} = t$, since if this sum is greater than $t$ then the second block of $\Sigma(Y_0)$ would not be all zeros, and if $\sum_\ell x_{i_\ell,j_\ell} < t$ then the encoding block of $\Sigma(Y_0)$ would not be~$t$.
\item $m = n$, by a similar argument in the counting block.
\item There is exactly one integer of the form $x^0_{i^*}$ in $Y_0$, for some $i^* \in \{1,\ldots,N\}$, as otherwise the most significant bit of $\Sigma(Y_0)$ would not be 1. 
\item $i^*=i_1 = \cdots = i_n$: Note that $x^0_{i^*}$ contributes $n \cdot (S_{\textup{max}}-s_{i^*})$ to the verification block of $\Sigma(Y_0)$, and so the remaining $n$ integers in $Y_0$ need to contribute together exactly $n \cdot s_{i^*}$ to this block, since the value of this block is $n \cdot S_{\textup{max}}$ in $t_0$. Since $S$ is an $n$-average-free set, the only way for this to occur is if all of these integers have $s_{i^*}$ encoded in their verification blocks, implying that $i^*=i_1 = \cdots = i_n$.
\end{itemize}
Let $i = i^*$ be the index in the last point above. Then $\sum_\ell x_{i,j_\ell} = t$ by the first point above, and so the subset  $\{x_{i,j_1},\ldots,x_{i,j_n}\}$ is a solution for the instance $(X_i,t_i)$. \qed
\end{proof}


\section{Scheduling Lower Bounds}
\label{sec:SLB}%

We next show how to apply Corollary~\ref{cor:thecor} to obtain $\Omega(n+P_{max} \cdot n^{1-\eps})$ lower bounds for several scheduling problems. In particular, we provide a complete proof of Theorem~\ref{thm:lowerbounds} in a sequence of lemmas below, each exhibiting a reduction from AND Subset Sum (or rather AND Partition) to the scheduling problem at hand.

In each reduction, we start with $N$ Partition instances $(X_1,t_1),\ldots,(X_N,t_N)$; these are Subset Sum instances with $t_i = \frac 12 \sum_{x \in X_i} x$. 
We write $s = \max_i |X_i|$ and $t = \max_i t_i$. 
We present reductions that transform these given instances into an instance $\mathcal{I}$ of a certain scheduling problem, such that $\mathcal{I}$ is a yes-instance if and only if $(X_i,t_i)$ is a yes-instance for all~$i$. 
The constructed instance $\mathcal{I}$ will consist of $n = O(Ns)$ jobs with maximum processing time $P_{max} = O(t)$. 
Since Corollary~\ref{cor:thecor} rules out time $\widetilde{O}(Ns + t (Ns)^{1-\eps})$ for AND Partition, it follows that the scheduling problem is not in time $\widetilde{O}(n + P_{max} \cdot n^{1-\eps})$, for any $\eps > 0$ assuming $\forall \exists$-SETH.

%

For an instance $(X_i,t_i)$, we let $x_{i,j}$ denote the $j$-th integer in~$X_i$. 

\subsection{Scheduling Notation and Terminology}

In all scheduling problems considered in this paper, we are given a set of jobs $J_1,\ldots,J_n$ to be scheduled non-preemptively on one or two identical parallel machines. Each job $J_j$ has a \emph{processing time} $p_j$, and according to the specific problem at hand, it may also have a \emph{due date} $d_j$, a \emph{release date} $r_j$, and a \emph{weight} $w_j$. We always use the same subscript for the job and its parameters. A \emph{schedule} consists of assigning each job $J_j$ a machine $M(J_j)$ and a \emph{starting time} $S_j \in \mathbb{N}^{\geq 0}$. The \emph{completion time} of job $j$ in a given schedule is $C_j = S_j + p_j$, and the \emph{makespan} of the schedule is its maximum completion time $C_{max} = \max_j C_j$. A schedule is \emph{feasible} if no two distinct jobs \emph{overlap} on the same machine; that is, for any pair of distinct jobs $J_j$ and $J_k$ with $M(J_j)=M(J_k)$ and $S_j \leq S_k$ we have $S_k \notin [S_j,C_j)$. Furthermore, when release dates are present, we require that $S_j \geq r_j$ for each job $J_j$. 

A job $J_j$ is said to be tardy in a given schedule if $C_j > d_j$, and otherwise it is said to be early. For each job $J_j$, we let $U_j \in \{0,1\}$ denote a Boolean variable with $U_j=1$ if $J_j$ is tardy and otherwise $U_j=0$. In this way, $\sum U_j$ denotes the number of tardy jobs in a given schedule, and $\sum w_jU_j$ denote their total weight. We let $T_j$ denote the \emph{tardiness} of a job $J_j$ defined by $T_j = \max\{0,C_j-d_j\}$, and we let $T_{max}=\max_j T_j$ denote the \emph{maximum tardiness} of the schedule. Below we use the standard three field notation $\alpha|\beta|\gamma$ introduced by Graham et al.~\cite{GrahamLLA79} to denote the various problems, where $\alpha$ denotes the machine model, $\beta$ denotes the constrains on the problem, and $\gamma$ is the objective function. Readers unfamiliar with the area of scheduling are also referred to~\cite{Pinedo2008} for additional background. 

\subsection{Problems on Two Machines}

We begin by considering scheduling problems on two parallel identical machines, as here our reductions are simpler to describe. Recall that in this setting, a schedule consists of assigning a starting-time $S_j$ and a machine $M(J_j)$ to each input job $J_j$. 

\subsubsection{\boldmath$P_2|\text{\emph{level-order}}|C_{max}$}

Perhaps the easiest application of Theorem~\ref{thm:ANDSubsetSum} is makespan minimization on two parallel machines when level-order precedence constraints are present~\cite{DolevWarmuth85,WangBellenguez19}. In this problem, jobs only have processing-times, and they are partitioned into classes $\mathcal{J}_1,\ldots,\mathcal{J}_k$ such that all jobs in any class $\mathcal{J}_i$ must be scheduled after all jobs in $\mathcal{J}_{i-1}$ are completed. The goal is to find a feasible schedule with minimum makespan $C_{max} = \max_j C_j$.

\begin{lemma}
\label{lem:level-order}%
$P_2|\text{level-order}|C_{max}$ has no $\widetilde{O}(n+P_{max} \cdot n^{1-\eps})$-time algorithm, for any $\eps > 0$, unless $\forall \exists$-SETH is false.  
\end{lemma}

\begin{proof}
First recall that a single Partition instance $(X,t)$ easily reduces to an instance of $P_2||C_{max}$ (\emph{i.e.} without precedence constraints on the jobs) by creating a job with processing time $x$ for each $x \in X$, and then setting the required makespan $C$ to be $C=t$. For reducing multiple Partition instances we can use the precedence constraints: For each instance $(X_i,t_i)$ of Partition, we create a class of jobs $\mathcal{J}_i$ which includes a job $J_{i,j}$ for each $x_{i,j} \in X_i$ with processing time $p_{i,j}=x_{i,j}$. Then since all jobs in class $\mathcal{J}_i$ must be processed after all jobs in $\mathcal{J}_1,\ldots,\mathcal{J}_{i-1}$ are completed, it is easy to see that the $P_2|\text{\emph{level-order}}|C_{max}$ instance has a feasible schedule with makespan at most $C= \sum_i t_i$ if and only if each Partition instance is a yes-instance. 

Indeed, if each $X_i$ has a subset $Y_i \subset X_i$ which sums up to $t_i = \frac{1}{2} \cdot \sum_j x_{i,j}$, then we can schedule all jobs $J_{i,j}$ associated with elements $x_{i,j} \in Y_i$ on the first machine (following all jobs associated with elements in $Y_1,\ldots,Y_{i-1}$), and all jobs $J_{i,j}$ associated with elements $x_{i,j} \notin Y_i$ on the second machine. This gives a feasible schedule with makespan at most $C$. Conversely, a schedule with makespan at most $C$ must have the last job in $\mathcal{J}_i$ complete no later than $\sum_{i_0 \leq i} t_{i_0}$, for each $i\in\{1,\ldots,N\}$. This in turn can only be done if each $X_i$ can be partitioned into two sets that sum up to $t_i$, which implies that each $(X_i,t_i)$ is a yes-instance. 

\smallskip
Starting from $N$ Partition instances $(X_1,t_1),\ldots,(X_N,t_N)$, each with at most $s$ integers and target at most $t$, our reduction constructs $n \le Ns$ jobs with maximum processing time $P_{max} \le t$. Therefore, any $\widetilde{O}(n+P_{max} \cdot n^{1-\eps})$-time algorithm for $P_2|\text{\emph{level-order}}|C_{max}$ would yield an $\widetilde{O}(Ns + t (Ns)^{1-\eps})$-time algorithm for AND Partition, which contradicts Corollary~\ref{cor:thecor}, assuming $\forall \exists$-SETH. \qed
\end{proof}

\subsubsection{\boldmath$P_2||T_{max}$ and \boldmath$P_2||\sum U_j$}

We next consider the $P_2||T_{max}$ and $P_2||\sum U_j$ problems, where jobs also have due dates, and the goal is to minimize the maximum tardiness and the total number of tardy jobs, respectively. The reduction here is very similar to the previous reduction. We create for each $i \in \{1,\ldots,N\}$, and each $x_{i,j}\in X_i$, a job $J_{i,j}$ with processing time $p_{i,j}=x_{i,j}$ and due date 
\[
d_{i,j}=d_i= \sum^i_{\ell=1} t_\ell = \frac{1}{2} \cdot \sum^i_{\ell=0} \sum_j x_{\ell,j}.
\]
Observe that for each $i \in \{1,\ldots,N\}$, all jobs $J_{i,j}$ can be scheduled early if and only if $X_i$ can be partitioned into two sets summing up to $t_i$. Thus, all jobs can be scheduled early if and only if all Partition instances are yes-instances. Note that this corresponds to both objective functions $T_{max}$ and $\sum U_j$ at value 0. Thus, using Corollary~\ref{cor:thecor} we obtain:
\begin{lemma}
\label{lem:TmaxSumUj}%
Both $P_2||T_{max}$ and $P_2||\sum U_j$ have no $\widetilde{O}(n+P_{max} \cdot n^{1-\eps})$-time algorithms, for any $\eps > 0$, assuming $\forall \exists$-SETH.  
\end{lemma}

\subsubsection{\boldmath$P_2|r_j \geq 0|C_{max}$}

Our final dual machine example is the problem of minimizing makespan when release dates are present, the classical $P_2|r_j \geq 0|C_{max}$ problem.
\begin{lemma}
$P_2|r_j \geq 0|C_{max}$ has no $\widetilde{O}(n+P_{max} \cdot n^{1-\eps})$-time algorithm, for any $\eps > 0$, unless $\forall \exists$-SETH is false.  
\end{lemma}

\begin{proof}
Let $(X_1,t_1),\ldots,(X_N,t_N)$ be $N$ instances of Partition. For each element $x_{i,j} \in X_i$ we create a job $J_{i,j}$ with processing time $p_{i,j}=x_{i,j}$ and release date $r_{i,j} = \sum_{\ell < i} t_\ell$. Note that there is a schedule for this instance with makespan $\sum_{i=1}^N t_i$, where each job is scheduled no earlier than its release date, if and only if each Partition instance is a yes-instance. Also note that the resulting instance has maximum processing time $P_{max} = \max_i t_i$ and total number of jobs $n \le N \cdot \max_i |X_i|$. As before, using Corollary~\ref{cor:thecor} we can now rule out time $\widetilde{O}(n+P_{max} \cdot n^{1-\eps})$, assuming $\forall \exists$-SETH. \qed
\end{proof}

\subsection{Problems on One Machine}

We next consider single machine problems. Obviously, a schedule in this case only needs to specify a starting time $S_j$ for each job $J_j$, and in case there are no release dates, a schedule can be simply thought of as a permutation of the jobs. 

\subsubsection{\boldmath$1||\sum w_jU_j$}

One of the most classical single-machine scheduling problems which already appeared in Karp's initial list of 21 NP-complete problems~\cite{Karp72} is the problem of minimizing the total weight of tardy jobs. Here each job $J_j$ has a due date $d_j$ and weight $w_j$, and the goal is to minimize $\sum w_jU_j$.  

\begin{lemma}
\label{lem:wjUj}%
Assuming $\forall \exists$-SETH, there is no $\widetilde{O}(n+P_{max} \cdot n^{1-\eps})$-time algorithm for $1||\sum w_jU_j$, for any $\eps > 0$.  
\end{lemma}

\begin{proof}
Let $(X_1,t_1),\ldots,(X_N,t_N)$ be $N$ instances of Partition. For each $i \in \{1,\ldots,N\}$, and for each $x_{i,j} \in X_i$, we create a job $J_{i,j}$ with the following parameters: 
\begin{itemize}
\item processing time $p_{i,j}=x_{i,j}$, 
\item weight $w_{i,j} = (N-i+1) \cdot x_{i,j}$,
\item and due date $d_{i,j} = d_i = \sum^i_{\ell=1}t_\ell$. 
\end{itemize}
We argue that there is a schedule for all jobs $J_{i,j}$ with total weight of tardy jobs at most $W=\sum^N_{i=1} (N-i+1) \cdot t_i$ if and only if each Partition instance $(X_i,t_i)$ is a yes-instance. 

Suppose that each $X_i$ has a subset $Y_i \subseteq X_i$ which sums up to $t_i$. Let $\mathcal{E}_i = \{J_{i,j} : x_{i,j} \in Y_i\}$ and $\mathcal{T}_i=\{J_{i,j}  : x_{i,j} \notin Y_i\}$ for $i\in\{1,\ldots,N\}$, and let $\mathcal{E} = \bigcup_i \mathcal{E}_i$ and $\mathcal{T}=\bigcup_i \mathcal{T}_i$. Then any schedule of the form $\mathcal{E}_1,\ldots,\mathcal{E}_N,\mathcal{T}$, where the order inside each subset of jobs is arbitrary, schedules all jobs in $\mathcal{E}$ early, and so the total weight of tardy jobs of such a schedule is at most the total weight of $\mathcal{T}$ which is 
$w(\mathcal{T}) = \sum_i w(\mathcal{T}_i) = \sum^N_{i=1} (N-i+1) \cdot t_i = W$.

Conversely, suppose there is a schedule for the jobs $J_{i,j}$ where the total weight of tardy jobs is at most $W$. Let $\mathcal{E}_{i}$ denote the set of early jobs in the schedule with due date $d_i$, for $i=\{1,\ldots,N\}$, and let $\mathcal{E}=\bigcup \mathcal{E}_i$. Then as the total weight of all jobs is $2W$, we have $w(\mathcal{E}) \geq W=\sum_i (N-i+1) \cdot t_i$. By our construction, this can only happen if we have $w(\mathcal{E}_i) \geq (N-i+1) \cdot t_i$ for each $i \in \{1,\ldots,N\}$, which in turn can only happen if $p(\mathcal{E}_i) \geq t_i$. Since all jobs in each $\mathcal{E}_i$ are early, we have $p(\mathcal{E}_i) \leq t_i$, and so $p(\mathcal{E}_i) = t_i$. It follows that for each $i \in \{1,\ldots,N\}$, the set $Y_i =\{ x_{i,j} : J_{i,j} \in \mathcal{E}_i\}=\{ p_{i,j} : J_{i,j} \in \mathcal{E}_i\}$ sums up to $t_i$. Thus we have found a solution for each Subset Sum instance $(X_i,t_i)$, and so the lemma follows.  \qed
\end{proof}

\subsubsection{\boldmath$1|Rej \leq R|\sum U_j$ and \boldmath$1|Rej \leq R|T_{max}$}

In scheduling with rejection problems~\cite{ShabtayGK13}, jobs $J_j$ are allowed not to be scheduled (\emph{i.e.} rejected) at the cost of $w_j$. Here we consider the case where the total cost of rejected jobs cannot exceed some prespecified bound $R$. Under this constraint, the $1|Rej \leq R|\sum U_j$ and $1|Rej \leq R|T_{max}$ problems focus on minimizing the number of tardy jobs $\sum U_j$ and the maximum tardiness of any job $T_{max}$, respectively. 

Note that there is a direct reduction from the $1||\sum w_jU_j$ problem to the $1|Rej \leq R|\sum U_j$ and $1|Rej \leq R|T_{max}$ problems: An instance of $1||\sum w_jU_j$ has a schedule with total weight at most $W$ if and only if there are jobs of total weight $R=W$ that can be rejected so that all remaining jobs can be scheduled early. Thus, the lemma below immediately follows from Lemma~\ref{lem:wjUj} above. 
\begin{lemma}
Assuming $\forall \exists$-SETH, both $1|Rej \leq R|\sum U_j$ and $1|Rej \leq R|T_{max}$ have no $\widetilde{O}(n+P_{max} \cdot n^{1-\eps})$-time algorithms, for any $\eps > 0$.  
\end{lemma}


\subsubsection{\boldmath$1|r_j \geq 0, Rej \leq R|C_{max}$}

In this problem, each job $J_j$ has a processing time $p_j$, a release date $r_j$, and a weight $w_j$, and the goal is to find a schedule that rejects jobs with total weight at most $R$ and minimizes the makespan of the remaining non-rejected jobs.

\begin{lemma}
There is no $\widetilde{O}(n+P_{max} \cdot n^{1-\eps})$-time algorithm for $1|r_j \geq 0, Rej \leq R|C_{max}$, for any $\eps > 0$, unless $\forall \exists$-SETH is false. 
\end{lemma}

\begin{proof}
Let $(X_1,t_1),\ldots,(X_N,t_N)$ be $N$ instances of Partition. For each $i \in \{1,\ldots,N\}$, and for each $x_{i,j} \in X_i$, we create a job $J_{i,j}$ with: 
\begin{itemize}
\item processing time $p_{i,j}=x_{i,j}$, 
\item weight $w_{i,j} = i \cdot x_{i,j}$,
\item and release date $r_{i,j} = r_i = \sum^{i-1}_{\ell=1}t_\ell$. 
\end{itemize}
We argue that there is a schedule for all jobs $J_{i,j}$ with makespan at most $C=\sum_i t_i$ that rejects jobs with cost at most $R=\sum_i i \cdot t_i$ if and only if each Partition instance $(X_i,t_i)$ is a yes-instance. 

Suppose that each $X_i$ has a subset $Y_i \subseteq X_i$ which sums up to $t_i$. Let $\mathcal{E}_i = \{J_{i,j} : x_{i,j} \in Y_i\}$ and $\mathcal{T}_i=\{J_{i,j}  : x_{i,j} \notin Y_i\}$ for $i\in\{1,\ldots,N\}$, and let $\mathcal{E} = \bigcup_i \mathcal{E}_i$ and $\mathcal{T}=\bigcup_i \mathcal{T}_i$. Then any schedule of the form $\mathcal{E}_1,\ldots,\mathcal{E}_N$, where the jobs in $\mathcal{T}$ are rejected, respects all release dates of jobs in $\mathcal{E}$, and has makespan $C_{max}=\sum_i t_i = C$. Moreover, the total cost of the rejected jobs is $w(\mathcal{T}) = \sum_i w(\mathcal{T}_i) = \sum^N_{i=1} i \cdot t_i = R$.

Conversely, suppose there is schedule for the jobs $J_{i,j}$ that respects all release dates, rejects jobs with weight at most $R$, and has makespan at most $C$. Let $\mathcal{E}_{i}$ denote the set of non-rejected jobs with release date $r_i$, for $i=\{1,\ldots,N\}$, and let $\mathcal{E}=\bigcup \mathcal{E}_i$. Then as the total weight of all jobs is $2R$, we have $w(\mathcal{E}) \leq R=\sum_i i \cdot t_i$. By our construction, this can only happen if we have $w(\mathcal{E}_i) \geq i \cdot t_i$ for each $i \in \{1,\ldots,N\}$, which in turn can only happen if $p(\mathcal{E}_i) \geq t_i$. On the other hand, the  release date $r_{i+1}$ of jobs in $\mathcal{E}_{i+1}$ can be respected only if $p(\mathcal{E}_i) \leq r_{i+1} = \sum_{\ell = 1}^i t_\ell$, and so $p(\mathcal{E}_i) = t_i$. It follows that for each $i \in \{1,\ldots,N\}$, the set $Y_i =\{ x_{i,j} : J_{i,j} \in \mathcal{E}_i\}=\{ p_{i,j} : J_{i,j} \in \mathcal{E}_i\}$ sums up to $t_i$. Thus, we have found a solution for each Subset Sum instance $(X_i,t_i)$, and so the lemma follows.  \qed
\end{proof}


\section{Conclusion} 

In this paper we considered the AND Subset Sum problem: Given $N$ instances of Subset Sum, determine whether all instances are yes-instances. We showed that the problem is essentially as hard as solving $N$ Subset Sum instances independently, and then used this result to strengthen existing lower bounds for several scheduling problems.  Our research is closely related to the question of whether Subset Sum on input $(X,t)$ can be solved in time $\widetilde{O}(\max_{x \in X} x + |X|)$, which is currently a central open problem in the area~\cite{AxiotisBJTW19,EisenbrandW18,GalilMargalit91,KoiliarisX17,Pisinger99}. Our results answer this question in the negative for several generalizations of Subset Sum. We believe that the line of thought in this paper can provide other results in a similar vein.

Observe that almost all scheduling problems considered in this paper do not have a matching upper-bound of $\widetilde{O}(P_{max} \cdot n)$ to the lower bound constructed in Section~\ref{sec:SLB}. The exception is $P_2|\text{level-order}|C_{max}$ which can be solved in time $O(P_{max} \cdot n)$ by using the known $O(P_{max} \cdot n)$-time Subset Sum algorithm~\cite{Pisinger99} (or the faster algorithms given in~\cite{Bring17,KoiliarisX17}) on each class of jobs $\mathcal{J}_i$ separately. It would be very interesting to close the gap for other problems listed in Theorem~\ref{thm:lowerbounds}. This could be done by either devising an $\widetilde{O}(P_{max} \cdot n)$-time algorithm for the problem, or by strengthening our lower bound mechanism.


\bibliographystyle{plain}
\bibliography{biblio}

\end{document}